\documentclass[journal,twoside]{ieeetran}
\usepackage{cite}
\usepackage{amsmath,amssymb,amsfonts}
\usepackage{amsthm}
\usepackage{algorithmic}
\usepackage{graphicx}
\usepackage{textcomp}
\def\BibTeX{{\rm B\kern-.05em{\sc i\kern-.025em b}\kern-.08em
    T\kern-.1667em\lower.7ex\hbox{E}\kern-.125emX}}
\title{I.A. Raptis: Observability of Linear Time-Invariant Systems with Relative Measurements: A Geometric Approach}

\theoremstyle{plain}
\newtheorem{thm}{\protect\theoremname}
\theoremstyle{plain}
\newtheorem{cor}[thm]{\protect\corollaryname}
\theoremstyle{plain}
\newtheorem{prop}[thm]{\protect\propositionname}
\theoremstyle{definition}
\newtheorem{defn}[thm]{\protect\definitionname}
\theoremstyle{plain}
\newtheorem{lem}[thm]{\protect\lemmaname}

\providecommand{\corollaryname}{Corollary}
\providecommand{\definitionname}{Definition}
\providecommand{\lemmaname}{Lemma}
\providecommand{\propositionname}{Proposition}
\providecommand{\theoremname}{Theorem}

\newcommand{\fatone}{\mathbf{1}}
\newcommand{\fatzero}{\mathbf{0}}

\newcommand{\real}{\mathbb{R}}

\begin{document}
\title{Observability of Linear Time-Invariant Systems with Relative Measurements: A Geometric Approach}
\author{Ioannis~A.~Raptis
\thanks{I. Raptis is with the Department
of Electrical and Computer Engineering, North Carolina Agricultural and Technical State University, Greensboro,
NC, 27411 USA e-mail: iraptis@ncat.edu}% <-this % stops a space
\thanks{A short, preliminary version of this paper is included in the Proceedings of the 2019 ASME Dynamic Systems and Control Conference \cite{raptis2019observability}.}% <-this % stops a space
}

\maketitle

\begin{abstract}
This paper explores the observability and estimation capability of
dynamical systems using predominantly relative measurements of the
system's state-space variables, with minimal to no reliance on absolute
measurements of these variables. We concentrate on linear time-invariant
systems, in which the observation matrix serves as the algebraic representation
of a graph object. This graph object encapsulates the availability
of relative measurements. Utilizing algebraic graph theory and abstract
linear algebra (geometric) tools, we establish a link between the
structure of the graph of relative measurements and the system-theoretic
observability subspace of linear systems. Special emphasis is given
to multi-agent networked systems whose dynamics are governed by the
linear consensus protocol. We demonstrate the importance of absolute
information and its placement to the system's dynamics in achieving
full-state estimation. Finally, the analysis shifts to the synthesis
of a distributed observer with relative measurements for single integrator
dynamics, exemplifying the relevance of the preceding analytical findings.
We support our theoretical analysis with numerical simulations.
\end{abstract}


\section{Introduction}
\IEEEPARstart{O}{ne} standard functionality of control systems is the estimation of
characteristic quantities of a system or process that changes with
time based on a series of measurements. This core problem involves
a computational unit that collects sensor data from the monitored
process and, in real-time, executes an algorithm to produce informed
estimates of relevant variables. Therefore, the culmination of every
estimation method for dynamical systems is an estimation algorithm:
a ``filter'' in stochastic systems, and an ``observer'' in deterministic
ones. Kalman's filter \cite{kalman1960new} is the most renowned in
this domain, while Luenberger laid the ground work for the design
of observers \cite{luenberger1964observing}. A powerful tool in the
synthesis of observers are the geometric methods \cite{hespanha2018linear,chen2012linear},
which leverage the linear algebraic properties of linear systems'
matrices to determine system theoretic properties and synthesize feedback
controllers, disturbance decouplers, and, state observers. Key milestones
in the field of geometric methods for linear systems have been marked
by Basile and Maro \cite{basile1969controlled}, Wonham \cite{wonham1974linear},
and, more recently, Trentelman \cite{trentelman2012control} in their
respective research. This rudimentary outline of the estimation challenge
gives rise to two correlated facets.

The first facet emerges when the sensor data offer relative, rather
than absolute, measurements \cite{Barooah2007}. A characteristic
example is the localization of a team of robots where the units exchange
relative measurements of bearings and range from their onboard exteroceptive
(e.g., laser camera) and interoceptive (e.g., IMU, encoders, etc.)
sensors. Similar problems arise in the formation of satellites \cite{Rahmani2008}
and clock synchronization in networks of processors \cite{barooah2007estimation}. 

The second facet arises when replacing the central processor, which
traditionally aggregates measurements, with a distributed network
of processor-sensor nodes. In this setup, each node (or agent) gathers
its own measurements and those from adjacent nodes to formulate a
local estimate of a system variable. This is the classical problem
of distributed estimation. It becomes particularly compelling when
the estimated system also consists of interacting dynamic nodes, e.g.,
the well-known agreement dynamics. A substantial portion of research
in distributed estimation focuses on analyzing and synthesizing distributed
versions of the Kalman filter \cite{olfati2007distributed,mahmoud2013distributed,olfati2005distributed,carli2008distributed,khan2008distributing}.
This involves examining the interaction between the dynamics of the
distributed filter (or node), the availability of local measurements,
and the nature of information exchanged with communicating nodes.
Distributed Luenberger observers \cite{kim2019completely,wang2017distributed,mitra2018distributed,park2016design}
are distinguished by adapting the classic observer concept to suit
the nuances of networked or multi-agent systems. Primary objectives
include ensuring the stability of the observer network and devising
methods to minimize the information required by the distributed observers.

The two subproblems, intriguingly, may intertwine. For instance, the
localization of a robotic swarm or mobile sensor network could occur
in a distributed fashion, rather than centrally. In this scenario,
each robot, or mobile sensor, computes its absolute position by relying
on the relative positions and information exchanged with neighboring
robots \cite{safavi2018distributed,langendoen2003distributed,roumeliotis2002distributed,piovan2013frame,martinez2007motion}.
This challenge embodies the integration of both facets: the relative
measurement aspect and the distributed estimation element. 

In this work, we delve into the observability of linear time-invariant
(LTI) dynamic systems, which predominantly rely on measurements of
differences between state variables, with few or no direct measurements
of the state variables themselves. The framework for representing
these relative observations is efficiently captured through a graph-based
model and its associated incidence matrix. We establish a set of mathematical
propositions that define necessary and sufficient conditions for the
observability of this specific category of dynamic systems. Further,
our analysis expands to include the observability of multi-agent dynamical
systems, which display a graph-associated structural form. We concentrate
on the agreement dynamics, specifically in scenarios where the graph
of the consensus protocol corresponds with that of the relative measurements.
We also provide a detailed exposition of how limited global measurements
influence the observability of the agreement dynamics. Additionally,
we propose a novel design for a distributed observer tailored for
a single-integrator multi-agent dynamical system using the system
theoretical properties previously delineated in the paper. 

The principal contribution of this work is a geometric analysis of
the observability problem using relative measurements, creating a
toolkit at the crossroads of linear algebra and algebraic graph theory.
This approach culminates in linking the spectral characteristics of
the graph, which represents the available relative measurements, with
specific subspaces generated by the state-space matrices of the dynamic
system. Through this linkage, we identify essential principles that
establish conditions for the system's well-understood concept of observability.
We demonstrate that geometric methods for analyzing systems with relative
measurements, and more broadly dynamical systems incorporating a structural
graph abstraction, enhance the role of \textquotedblleft structure\textquotedblright{}
in mathematical analysis. This focus significantly reduces reliance
on matrix manipulation. Furthermore, this methodology facilitates
the ergonomic synthesis of distributed observers: rather than immediately
pursuing an observer design that conforms to traditional frameworks,
we can initially assess solvability and existence based on the system's
dynamics and the measurement graph, subsequently constructing the
observer if feasible. In essence, our goal is to develop a method
that prioritizes synthesis---focusing on structure---over design,
which typically involves determining numerical values for free parameters.

The structure of the paper is as follows: Section \ref{sec:MathPrelim}
provides mathematical preliminaries and notation pertinent to the
problem. Section \ref{sec:Problem-Statement} presents a detailed
description of the problem addressed in this paper. In Section \ref{sec:SpecialMatrices},
we delve into a series of propositions concerning the observability
of LTI systems monitored by relative measurement models associated
with specific graphs. The core contribution is articulated in Section
\ref{sec:GeneralRSS}, which delineates conditions for the observability
of generic LTI systems and the agreement dynamics based on relative
observations of a general structure. Conditions for observability
of the agreement dynamics with a single absolute measurement are explored
in Section \ref{sec:SingleMeas}. The methodology for designing a
distributed observer for single-integrator multi-agent systems is
outlined in Section \ref{sec:Distributed-Estimators}. Finally, Section
\ref{sec:Simulations} includes validating numerical simulations to
support the theoretical findings.

\section{Notation, Terminology, and Mathematical Preliminaries\label{sec:MathPrelim}}
This section succinctly revisits essential concepts and notations
from linear algebra, matrix analysis, and algebraic graph theory,
which are integral to our subsequent discussions. Readers may consult
standard linear algebra texts such as Axler \cite{axler2015}, Horn
and Johnson \cite{horn2012}, and Garcia and Horn \cite{garcia2017second},
graph theory references like Godsil and Royle \cite{godsil2013algebraic},
Bollobás \cite{bollobas2013}, and foundational works on networked
control systems including Mesbahi and Egerstedt \cite{mesbahi2010graph},
Bullo et al. \cite{bullo2009distributed}, and Olfadi-Saber et al.
\cite{olfati2007consensus}.

\textit{Notational Conventions:} The set $\mathbb{R}^{n}$ denotes
the n-dimensional real coordinate space, comprising all n-tuples of
real numbers. A column vector $x\in\mathbb{R}^{n}$ with components
$x_{1},x_{2},\ldots,x_{n}$ is represented as $x=[x_{1};x_{2};\ldots;x_{n}]$.
The transpose of the column vector $x$, denoted by $x^{T}$, is expressed
as $[x_{1},x_{2},\ldots,x_{n}]$. The symbol $\mathbf{0}_{n}$ signifies
an $n$-dimensional column vector of zeros, while $\mathbf{1}_{n}$
indicates a similar vector but with ones. The $n\times n$ zero and
identity matrices are denoted by $O_{n}$ and $I_{n}$, respectively.
The notation $|\cdot|$ refers to either the cardinality of a set
or the absolute value of a real number, contextually.

\textit{Linear Algebra:} Consider a set of vectors $V=\{v_{1},v_{2},\ldots,v_{m}\}$
in a finite-dimensional vector space $\mathcal{V}$. The span of $V$,
denoted by $\text{span}(V)$, is the set of all possible linear combinations
of the vectors in $V$ and forms a subspace of $\mathcal{V}$, i.e.,
$\text{span}(V)\subseteq\mathcal{V}$. The set $V$ is a basis of
$\mathcal{V}$ if and only if its elements are linearly independent
and its span covers $\mathcal{V}$. In any finite-dimensional vector
space, all bases have the same cardinality, defining the dimension
of the space, denoted as $\dim(\mathcal{V})$.

\textit{Matrix Algebra:} An $m\times n$ real matrix $A$ defines
three fundamental subspaces in $\mathbb{R}^{n}$: the column space
(or image), the row space (or coimage), and the nullspace (or kernel)
of $A$, denoted by $\text{null}(A)$. For $x\in\text{null}(A)$,
it holds that $Ax=\mathbf{0}_{n}$. The rank of $A$, $\text{rank}(A)$,
is the dimension of its column space, while its nullity, $\text{nullity}(A)$,
is the dimension of its nullspace. The rank-nullity theorem asserts
that for $A$, $\text{rank}(A)+\text{nullity}(A)=n$. In the case
of a square matrix $A$ with eigenvalue $\lambda$, the associated
eigenspace in $\mathbb{R}^{n}$ is the span of the eigenvectors corresponding
to $\lambda$.

\textit{Graph Theory:} An undirected graph is a pair $\mathcal{G}=(\mathcal{V},\mathcal{E})$,
with $\mathcal{V}$ representing vertices and $\mathcal{E}$, a set
of unordered vertex pairs or edges. In a directed graph, edges are
ordered pairs. Vertices $i$ and $j$ connected by an edge $(i,j)\in\mathcal{E}$
are adjacent or neighboring. The neighborhood $\mathcal{N}(i)$ of
vertex $i$ comprises its adjacent vertices, and its degree, $d(i)$,
is the count of such vertices. A graph is connected if every pair
of vertices is linked by a path; otherwise, it is disconnected.

\textit{Algebraic Graph Theory:} This field links graphs with algebraic
structures like matrices. For an undirected graph $\mathcal{G}$ with
$n$ vertices, the adjacency matrix $\mathcal{A}(\mathcal{G})$ is
defined such that $\mathcal{A}(\mathcal{G})_{ij}=1$ if vertices $i$
and $j$ are connected, and $0$ otherwise. The degree matrix $\Delta(\mathcal{G})$
has diagonal entries $\Delta(\mathcal{G})_{ii}=d(i)$ and zeros elsewhere.
The graph Laplacian, $\mathcal{L}$, is $\Delta(\mathcal{G})-\mathcal{A}(\mathcal{G})$.
The Laplacian's rows sum to zero, and it is symmetric and positive
semidefinite. The eigenvalue $\lambda_{1}(\mathcal{G})=0$ corresponds
to the eigenvector $\mathbf{1}_{n}$, placing it in $\mathcal{L}$'s
nullspace. For a connected graph, the nullspace is one-dimensional,
and $\lambda_{2}(\mathcal{G})$ is the smallest positive eigenvalue.

The incidence matrix $\mathcal{D}(\mathcal{G})$ for a directed graph
with $\epsilon$ edges is defined with elements
\[
\mathcal{D}(\mathcal{G})_{ij}=\begin{cases}
-1, & \text{if vertex \ensuremath{i} is the tail of edge \ensuremath{j};}\\
\quad1, & \text{if vertex \ensuremath{i} is the head of edge \ensuremath{j};}\\
\quad0, & \text{otherwise}.
\end{cases}
\]

This matrix is also applicable to undirected graphs given an arbitrary
orientation. Its column sums are zero, reflecting the graph's topology. 

\section{Problem Statement\label{sec:Problem-Statement}}
We examine the observability of Linear Time-Invariant (LTI) systems
predominantly through relative measurements. For ease of reference,
we term this observation setup a Relative Sensing System (RSS). In
RSS, the nodes are equipped to measure the differences between pairs
of state variables, termed \emph{relative measurements}, along with
certain state variables, known as \emph{absolute} or \emph{anchor
measurements}, of an LTI system. Our primary aim is to determine the
conditions under which the LTI system is observable when the available
observations mainly consist of relative measurements of the system's
states, complemented by a minimal number (potentially as few as one)
or even in the absence of anchor measurements. 

To start formulating the problem at hand, consider the customary LTI
system
\begin{align}
\dot{x} & =Ax+Bu\label{eq:LTI-1}\\
y & =Cx,\nonumber 
\end{align}
where $x\in\real^{n}$ is the LTI system's state vector
$\left[x_{1};x_{2};\ldots;x_{n}\right]$, $u\in\real^{l}$ is its
input vector $\left[u_{1};u_{2};\ldots;u_{l}\right]$, and $y\in\real^{m}$
is its observation vector $\left[y_{1};y_{2};\ldots;y_{m}\right]$.
The matrices $A\in\real^{n\times n}$, $B\in\real^{n\times l}$, and
$C\in\real^{m\times n}$ are the state, input, and observation matrices
of the LTI system, respectively. From standard results, the observability
matrix, $\mathcal{O}$, of the LTI system (\ref{eq:LTI-1}) is defined
by
\begin{equation}
\mathcal{O}=\left[\begin{array}{c}
C\\
CA\\
CA^{2}\\
\vdots\\
CA^{n-1}
\end{array}\right].\label{eq:observability_matrix}
\end{equation}

For brevity, the LTI system (\ref{eq:LTI-1}), when we address its
observability property, is considered equivalent to the matrix pair
$(A,C)$; that is, if the system (\ref{eq:LTI-1}) is observable,
it is the same to say that the pair $(A,C)$ is observable. 

The available relative measurements of the RSS are conveniently represented
by a directed graph $\mathcal{G}=\left(\mathcal{V},\mathcal{E}\right)$.
The node set $\mathcal{V}$ consists of the states of the LTI system,
and the edge set $\mathcal{E}$ represents the availability of relative
measurements. For example, the edge $(i,j)\in\mathcal{E}$ between
the nodes $i$ and $j$ states that sensing node $i$ can measure
the difference $x_{j}-x_{i}$, and vice versa. When this is the case,
the incidence matrix, $\mathcal{D}$, of graph $\mathcal{G}$ is used
to capture all the relative measurements of the RSS with the vector
$\mathcal{D}^{T}x$. With this representation, the graph $\mathcal{G}$
determines the elements of the observation matrix $C$. We say that
the LTI system (\ref{eq:LTI-1}) is \emph{monitored} by an RSS; we
refer to the graph $\mathcal{G}$ ---which captures the relative
measurements of the RSS--- as the \emph{relative observations graph},
matrix $\mathcal{D}^{T}$ as the\emph{ relative observations matrix},
and vector $\mathcal{D}^{T}x$ as the \emph{relative observation vector}.
It is noted that the orientation of the graph does not impact the
subsequent analysis; hence, the signs of the incident matrix elements,
can be arbitrarily chosen. 

Before proceeding to the main results of this paper, we outline a
number of preparatory theorems related to the observability of LTI
systems. About the notation, consider the generic representation of
the LTI system (\ref{eq:LTI-1}). The first theorem involves the standard
definition of observability. 
\begin{thm}
\label{thm:observability-matrix}The matrix pair $(A,C)$ is observable
if and only if the rank of the $nm\times n$ observability matrix
$\mathcal{O}$ in (\ref{eq:observability_matrix}) is $n$.
\end{thm}
\begin{proof}
The proof of this basic theorem can be found in any textbook of linear
systems theory, for example in \cite[pg.144]{hespanha2018linear}.
\end{proof}
The following propositions associate the eigenvectors of the state
matrix $A$ and the rank of the observation matrix $C$ with the observability
of the LTI system, and will be used extensively in the succeeding
analysis.
\begin{thm}
\label{thm:right_eigen}The matrix pair $(A,C)$ is not observable
if and only if there exists an eigenvector $\nu\in\real^{n}$ of $A$
such that 
\[
C\nu=\fatzero_{m}.
\]
\end{thm}
\begin{proof}
The proof of Theorem \ref{thm:right_eigen} is detailed in \cite[pg. 319]{bay1999fundamentals}.
\end{proof}
A useful restatement of the above theorem is the subsequent corollary:
\begin{cor}
\label{cor:disjoint_eigenspaces}The matrix pair $(A,C)$ is observable
if none of the eigenvectors of $A$ belong to the kernel of $C$, $null(C)$,
or equivalently if $null(sI_{n}-A)\cap null(C)=\emptyset$ for every
$s\in\real$.
\end{cor}
\begin{proof}
Theorem \ref{thm:right_eigen} can be rephrased as the following logical
implication: the pair $(A,C)$ is observable if there is no eigenvector,
$\nu$, of $A$ for which $C\nu=\fatzero_{m}$. Hence, the pair $(A,C)$
is observable if there is no eigenvector of $A$ that belongs to the
kernel of $C$. In other words, the eigenspace of $A$ and the kernel
of $C$ are disjoint. Let $\lambda$ be an eigenvalue of $A$; the
eigenspace of $A$ can can be conveniently expressed as the kernel
of the matrix $\lambda I_{n}-A$; therefore, observability of $(A,C)$
requires $null(sI_{n}-A)\cap null(C)=\emptyset$ for every $s\in\real$.
A very similar derivation of this proof can be found in \cite{hespanha2018linear}. 
\end{proof}
A first consequence of the above corollary is the following:
\begin{prop}
\label{prop:common-kernels-unobservable}The matrix pair $(A,C)$
is unobservable if an eigenvector belongs to both kernels of $A$
and $B$, or, alternatively, if the two matrices have intersecting
kernels, i.e., $null(A)\cap null(C)\neq\emptyset$ .
\end{prop}
\begin{proof}
This proposition can be viewed as an instance of Corollary \ref{cor:disjoint_eigenspaces}
when zero is an eigenvalue of matrix $A$. Suppose that $v$ is a
vector belonging to the intersection of matrices $A$ and $C$ kernels.
We start with the observation that vector $v$ also belongs to the
eigenspace of $A$ with corresponding eigenvalue zero, because $Av=\fatzero_{n}=0v$.
Since $v\in null(C)$, we have $Cv=\fatzero_{m}$, therefore the pair
$(A,C)$ is unobservable on the account of Theorem \ref{thm:right_eigen}
because an eigenvector of $A$ belongs to the kernel of $C$.
\end{proof}
Another consequence of Theorem \ref{thm:right_eigen} is the following
proposition:
\begin{prop}
\label{prop:C-rank-observable}The matrix pair $(A,C)$ is observable
if the rank of the observation matrix $C$ is $n$.
\end{prop}
\begin{proof}
If the matrix $C$ is of rank $n$, then by virtue of the rank-nullity
theorem the nullity of $C$ is zero; hence, there does not exist a
vector $v\in\real^{n}$ such that $Cv=0$. Thereby, in light of Theorem
\ref{thm:right_eigen}, the LTI system (\ref{eq:LTI-1}) is observable.
\end{proof}
The Popov-Belevitch-Hautus test, commonly known as the PBH or Hautus
test, for observability is an elegant restatement of the eigenvector
test of Theorem \ref{thm:right_eigen} and the rank conditions of
Corollary \ref{cor:disjoint_eigenspaces} and Proposition \ref{prop:C-rank-observable}. 
\begin{thm}
(Popov-Belevitch-Hautus test for observability) The LTI system given
in (\ref{eq:LTI-1}) is observable if and only if
\[
rank(\left[\begin{array}{c}
sI_{n}-A\\
\hline C
\end{array}\right])=n,\text{ for every }s\in\real.
\]
\end{thm}
\begin{proof}
The proof of the PBH test for observability can be found in \cite[pg.321]{bay1999fundamentals}
or \cite[pg.145]{hespanha2018linear}.
\end{proof}
In summation, the PBH test establishes the equivalence that the pair
$(A,C)$ is observable if and only if $null(sI_{n}-A)\cap null(C)=\emptyset$
for every $s\in\real$.

\section{Observability of Special RSS Observation Matrices\label{sec:SpecialMatrices}}
In this section, we investigate the observability property of the
generic LTI system given in (\ref{eq:LTI-1}) when the observation
matrix $C$ corresponds to an RSS with distinctive topology. Our aim
is to gradually bring out key notions that are needed in the subsequent
analysis of more general cases. Consider an RSS consisting of a collection
of ``circular'' relative measurements with an observation matrix
that corresponds to a cycle relative observations graph $\mathcal{G}$
having $n$ vertices and edges, such that $y_{i}=x_{i}-x_{i+1}$,
with $i=1,\ldots,n-1$, and $y_{n}=x_{n}-x_{1}$.
\begin{defn}
The \emph{$n\times n$ cycle observation matrix} $C_{n}^{c}$ is defined
as the difference $C_{n}^{c}=I_{n}-R_{n}$, where $R_{n}$ is the
$n\times n$ \emph{full cycle shift matrix} given by
\[
R_{n}=\begin{bmatrix}0 & 1 & 0 & \cdots & 0\\
0 & 0 & 1 & \cdots & 0\\
\vdots & \vdots & \vdots & \cdots & 0\\
0 & 0 & 0 & \cdots & 1\\
1 & 0 & 0 & \cdots & 0
\end{bmatrix}.
\]

At this point, we narrow our attention to the cycle observation matrix,
$C_{n}^{c}$, because it celebrates a number of featured properties
such as connectivity, presence of a cycle, and symmetry. The following
proposition specifies the rank of the $n\times n$ cycle observation
matrix $C_{n}^{c}$:
\end{defn}
\begin{prop}
\label{prop:Cs} The rank of the $n\times n$ cycle observation matrix
$C_{n}^{c}$ is $n-1$. 
\end{prop}
\begin{proof}
Assume that the vector $\nu=\left[\nu_{1};\nu_{2};\cdots;\nu_{n}\right]$
belongs to the nullspace of matrix $C_{n}^{c}$, i.e., $\nu\in null(C_{n}^{c})$;
then, for the vector $\nu$ we have 
\[
C_{n}^{c}\nu=\fatzero_{n}\Rightarrow\begin{bmatrix}\nu_{1}-\nu_{2}\\
\nu_{2}-\nu_{3}\\
\vdots\\
\nu_{n}-\nu_{1}
\end{bmatrix}=\fatzero_{n}\Rightarrow\nu_{1}=\nu_{2}=\cdots=\nu_{n}.
\]
From the above right-hand side equalities, it holds that $C_{n}^{n}\fatone_{n}=\fatzero_{n}$.
Therefore, the nullspace of $C_{s}$ is the spanning set of $\fatone_{n}$,
i.e., $null\left(C_{n}^{c}\right)=span\left(\fatone_{n}\right)$,
and its nullity is one ($nullity(C_{n}^{c})=1$). From the rank-nullity
Theorem, we get 
\[
rank\left(C_{n}^{c}\right)+nullity\left(C_{n}^{c}\right)=n\Rightarrow rank\left(C_{n}^{c}\right)=n-1.
\]
\end{proof}
Intuitively, having available only relative measurements of the states
may limit the observability of the LTI system. This claim is somewhat
expected because Proposition \ref{prop:Cs} has established that the cycle
observation matrix, $C_{n}^{c}$, is rank deficient.  

For example, in the trivial case where $A=0$ (absence of dynamics)
and $u=0$, such that $\dot{x}=0$, it is not possible to obtain an
estimate of $x$ from the observation vector $y=C_{n}^{c}x$, where
absolute measurements are entirely missing. This perceptible notion
is formally asserted in the following Theorem:
\begin{thm}
\label{thm:Cs_observability} The matrix pair ($A,C_{n}^{c}$) is
observable if and only if no eigenvector of matrix $A$ belongs to
the span of $\fatone_{n}$. 
\end{thm}
\begin{proof}
The PBH test requires for the LTI system (\ref{eq:LTI-1}) to be observable
that the following rank condition will hold: 
\[
rank(\left[\begin{array}{c}
sI_{n}-A\\
\hline C_{n}^{c}
\end{array}\right])=n,\text{ for every }s\in\real.
\]

Because the compound matrix of the PBH test has $n$ columns, the
requirement for observability is that its nullspace is empty. Another
way to view this requirement is that the nullspace of the matrix $\left[sI-A\right]$
with the nullspace of $C_{n}^{c}$ are disjoint, i.e., $null(sI-A)\cap null(C_{s})=\emptyset$,
for every $s\in\real$, or that there can be no eigenvector of $A$
in the nullspace of $C_{n}^{c}$. 

In Theorem \ref{prop:Cs}, we have shown that the one-dimensional
nullspace of $C_{n}^{c}$ is the span of the vector $\fatone_{n}$.
Thus, the LTI system is observable if and only if no eigenvector of
$A$ belongs to $span\left(\fatone_{n}\right)$; that is, $A\fatone_{n}\neq s\fatone_{n}$
for any $s\in\real$.
\end{proof}
As an example, consider the cycle observation matrix $C_{3}^{c}$
that corresponds to an RSS with three nodes and has the form
\[
C_{3}^{c}=\left[\begin{array}{ccc}
1 & -1 & 0\\
0 & 1 & -1\\
-1 & 0 & 1
\end{array}\right].
\]
The eigenvectors of the diagonal state matrix $A=diag(1,1,2)$ are
$v_{1}=[1;0;0]$, $v_{2}=[0;1;0]$, and $v_{3}=[0;0;1]$; where none
of which belongs to the span of $\fatone_{3}$. As expected, the rank
of the observability matrix is three; thus, the LTI system is observable.
Now, assume that the state matrix, $A$, is the identity matrix $I_{3}$.
It is well-known that $I_{3}$ has a single eigenvalue $\lambda=1$,
and its corresponding eigenspace is the entire $\real^{3}$. Because
$\fatone_{3}$ belongs to the eigenspace of $A$, the system is unobservable
and its observability matrix is rank defficient (i.e., $rank\left(\mathcal{O}\right)=2$).
Returning to the motivating example of Theorem \ref{thm:Cs_observability},
the state matrix $A=O_{3}$ has one eigenvalue $\lambda=0$, and its
corresponding eigenspace is the entire $\real^{3}$; hence, the system
that corresponds to the matrix pair $\left(O_{3},C_{3}^{g}\right)$
is unobservable and, similarly to the previous case, $rank\left(\mathcal{O}\right)=2$.

A mean to guarantee the observability of an LTI system that is monitored
by an RSS with observation matrix $C_{n}^{c}$ is to add an anchor
node to the observation vector, e.g., without loss of generality,
the absolute measurement of the state $x_{n}$. Therefore, we define
the $\left(n+1\right)\times n$ \emph{anchored cycle observation matrix}
$C_{n}^{a}$ as follows:
\[
C_{n}^{a}=\left[\begin{array}{c}
C_{n}^{a}\\
\hline e_{n}
\end{array}\right]\text{, where }e_{n}=\left[0,0,\ldots0,1\right].
\]
\begin{prop}
\label{proposition:Ca_rank}The rank of the $\left(n+1\right)\times n$
anchored cycle observation matrix $C_{n}^{a}$ is $n$. 
\end{prop}
\begin{proof}
It suffices to show that $nullity\left(C_{n}^{a}\right)$=0 or $null\left(C_{n}^{a}\right)=\emptyset$.
Assume that the vector $\nu=\left[\nu_{1};\nu;\ldots;\nu_{n}\right]$
belongs to the nullspace of matrix $C_{n}^{a}$, i.e., $\nu\in null(C_{g})$.
Then, for vector $\nu$ we have the following series of implications: 
\begin{align*}
C_{n}^{a}\cdot\nu=\fatzero_{\left(n+1\right)} & \Rightarrow\begin{bmatrix}\nu_{1}-\nu_{2}\\
\nu_{2}-\nu_{3}\\
\vdots\\
\nu_{n}-\nu_{1}\\
\nu_{n}
\end{bmatrix}=\fatzero_{\left(n+1\right)}\\
 & \Rightarrow\nu_{1}=\nu_{2}=\cdots=\nu_{n}=0.
\end{align*}

Therefore, only the trivial eigenvector, $\fatzero_{n}$, belongs
to the nullspace of matrix $C_{n}^{a}$, establishing that $nullity\left(C_{n}^{a}\right)=0$.
From the rank-nullity theorem, we get
\[
rank\left(C_{n}^{a}\right)+nullity\left(C_{n}^{a}\right)=n\Rightarrow rank\left(C_{n}^{a}\right)=n
\]
\end{proof}
In light of Proposition \ref{prop:C-rank-observable}, Proposition
\ref{proposition:Ca_rank} constitutes every LTI system with observation
matrix $C_{n}^{a}$ observable. In summary, this section examines
the observability property of LTI systems that are monitored by an
RSS with an observation matrix that corresponds to a cycle graph with
and without an anchor node. 

\section{Observability of General RSSs\label{sec:GeneralRSS}}
In this section, our attention is shifted to the overall instance
of LTI systems monitored by an RSS, which does not necessarily have
a distinguishing topology. Consider the standard description of LTI
systems given in (\ref{eq:LTI-1}). The RSS is mathematically encapsulated
by the observation vector $y=\mathcal{D}^{T}x$, corresponding to
an observation graph $\mathcal{G}$. 

In what follows, unless stated otherwise, we consider the graph $\mathcal{G}$
be connected with $n$ vertices and $\epsilon$ edges. It is silently
implied that the incidence matrix of $\mathcal{G}$ is $\mathcal{D}$.
The orientation of the graph is assumed arbitrary, as it is immaterial
to the derivation of the subsequent propositions. In the following
analysis, we will use frequently the premises of the next two lemmas.
\begin{lem}
\label{lem:rank-Dt}The transpose of its incidence matrix, $\mathcal{D}^{T}$,
has rank $n-1$.
\end{lem}
\begin{proof}
The proof of this lemma can be found in \cite[pg.24]{biggs1993algebraic}
or \cite[pg.166]{godsil2013algebraic}.
\end{proof}
\begin{lem}
\label{lem:kernel-of-Dt}The kernel of the transpose of the incidence
matrix, $\mathcal{D}^{T}$, is the span of $\fatone_{n}$, that is,
$null(\mathcal{D}^{T})=span(\fatone_{n})$ and in consequence $nullity(\mathcal{D}^{T})=1$.
\end{lem}
\begin{proof}
We prove this lemma by contradiction. Assume that vector $v=[v_{1};v_{2},\ldots;v_{n}]$
does not belong to the $span(\fatone_{n})$ but it belongs to the
kernel of $\mathcal{D}^{T}$: by assumption, we can write 
\[
\mathcal{D}^{T}v=\fatzero_{n}.
\]
Then, for every element pair $(i,j)$ of the edge set $\mathcal{E}$
of $\mathcal{G}$, we have $v_{i}-v_{j}=0$. Because $\mathcal{G}$
is connected, this implies that $v_{i}=v_{j}$ for all $i,j\in\mathcal{V}$.
Because all elements of $v$ are equal, it belongs to the span of
$\fatone_{n}$. It is trivial to show that $span(\fatone_{n})$ is
a single-dimension vector space, therefore, $nullity(\mathcal{D}^{T})=1$,
which concludes the proof.
\end{proof}
The observability of the LTI system (\ref{eq:LTI-1}) monitored by
an RSS with observation vector $y=\mathcal{D}^{T}x$ is established
by our next theorem. 
\begin{thm}
\label{thm:observability_LTI_Dt}The matrix pair ($A$, $\mathcal{D}^{T}$)
is observable if and only if no eigenvector of $A$ belongs to the
kernel of $\mathcal{D}^{T}$, namely, $span\left(\fatone_{n}\right)$. 
\end{thm}
\begin{proof}
The proof is derived in much the same way as Theorem \ref{thm:Cs_observability}.
The pair ($A$, $\mathcal{D}^{T}$) is observable if the kernels of
$sI_{n}-A$, for every $s\in\real$, and $\mathcal{D}^{T}$ are disjoint,
i.e., $null(sI_{n}-A)\cap null(\mathcal{D}^{T})=\emptyset$. Lemma
\ref{lem:kernel-of-Dt} asserts that the kernel of $\mathcal{D}^{T}$
is the span of $\fatone_{n}$. It follows that the matrix pair ($A$,
$\mathcal{D}^{T}$) is observable if there is no eigenvector of $A$
that belongs to $span(\fatone_{n})$, or there is no real $s$
such that $A\fatone_{n}=s\fatone_{n}$.
\end{proof}
Analogous to the case of the cycle observation matrix, $C_{n}^{c}$,
the following lemma manifests that the observability of an LTI system
monitored by a connected RSS
is assured if an anchor node is augmented in the relative measurements
vector $y=\mathcal{D}^{T}x$. There is no loss of generality in assuming
that the measurement of the state $x_{n}$ is the anchor node of the
connected RSS. 
\begin{lem}
\label{lem:observability_Dt_en_generic}The observation matrix 
\[
\left[\begin{array}{c}
\mathcal{D}^{T}\\
\hline e_{n}
\end{array}\right]
\]
is full rank and renders the LTI system (\ref{eq:LTI-1})
observable.
\end{lem}
\begin{proof}
In like manner with previous proofs, in order to show that
\[
rank(\left[\begin{array}{c}
\mathcal{D}^{T}\\
\hline e_{n}
\end{array}\right])=n,
\]
we only need to show that the nullspaces of $\mathcal{D}^{T}$
and $e_{n}$ are disjoint. If this is the case, then the nullity of
the compound observation matrix is zero, and the compound matrix is
full column rank. By Lemma \ref{lem:kernel-of-Dt} it holds $null(\mathcal{D}^{T})=span(\fatone_{n})$,
and all vectors of $null(\mathcal{D}^{T})$ have the form $\alpha\fatone_{n}$
where $\alpha$ is a nonzero real number. Because $e_{n}(\alpha\fatone_{n})\neq0$
for every $\alpha$, there is no element of $null(\mathcal{D}^{T})$
that belongs to $null(e_{n})$. Hence, the two nullspaces are disjoint
and the rank of the compound observation matrix is $n$. Finally,
by Proposition \ref{prop:C-rank-observable} because the observation
matrix is full rank, the instance of the LTI system (\ref{eq:LTI-1})
is observable, and the proof is complete.
\end{proof}
The importance of Lemma \ref{lem:observability_Dt_en_generic} is
the assertion that observability of an LTI system is separate from
the system's dynamics ---that is, independent of the state matrix
$A$--- because the observations map consisting of relative measurements
(of a connected RSS graph $\mathcal{G}$) and an anchor node is full-row
rank. 

We can conveniently extend the postulates of the previous propositions
to the analysis of multi-agent dynamical systems, viz., systems where
their dynamic structure is associated with the topology of a graph
object. The most prominent example of multi-agent coordination is
the consensus protocol (or the agreement dynamics). Here, we investigate
the observability of the agreement-dynamics of a multi-agent dynamical
monitored by an RSS. In the reminder of this section, we account for
the case that the interconnections of the agreement dynamics and the
RSS are represented by the same graph $\mathcal{G}$, where $\mathcal{L}$
and $\mathcal{D}$ are its Laplacian and incidence matrices, respectively.
\begin{cor}
The agreement dynamics 
\begin{align*}
\dot{x} & =-\mathcal{L}x+u\\
y & =\mathcal{D}^{T}x
\end{align*}
are unobservable. 
\end{cor}
\begin{proof}
The proof is straightforward. From existing knowledge, the kernel
of the Laplacian matrix $\mathcal{L}$, of a connected graph with
$n$ vertices, is $span(\fatone_{n})$ (see \cite[pg.279]{godsil2013algebraic}
and \cite[pg.27]{mesbahi2010graph}). It follows that the matrices
$\mathcal{L}$ and $\mathcal{G}$ have the same kernel; therefore,
the matrix pair $\left(-\mathcal{L},\mathcal{D}^{T}\right)$ is unobservable
by Theorem \ref{thm:observability_LTI_Dt}. 
\end{proof}
According to Lemma \ref{lem:observability_Dt_en_generic}, if at least
one anchor node is added to the observation vector of the RSS, the
matrix pair $\left(-\mathcal{L},\mathcal{D}^{T}\right)$ becomes observable. 

Hitherto, the centerpiece of our analysis has been the relative observation
matrix, $\mathcal{D}^{T}$, and vector, $\mathcal{D}^{T}x$, of an
RSS's observation graph $\mathcal{G}$. A distinctive observation
matrix of like importance is the Laplacian matrix, $\mathcal{L},$
of the graph $\mathcal{G}$. It is reminded that the two matrices
are associated by the relationship $\mathcal{L}=\mathcal{D}\mathcal{D^{T}}$.
We have already explained that the relative observation vector captures
all the available relative measurements of the RSS. Let us define
$y_{i}^{j}=x_{j}-x_{i}$ as the relative measurement of sensing node
$j$ with respect to node $i$; and $\mathcal{N}_{i}$ the set of
adjacent (neighboring) sensing nodes to $i$, determined by the graph
$\mathcal{G}$. For the observation vector $\mathcal{L}x$, after
some elementary algebra, one can show that its $i$-th element is 
\[
\sum_{j\in\mathcal{N}_{i}}y_{i}^{j}=\sum_{j\in\mathcal{N}_{i}}\left(x_{j}-x_{i}\right),
\]
which is the sum of relative information that the sensing node $i$
receives from its neighbors, and it is the observation (or output)
that the sensing node $i$ has available for decision-making (estimation
or control). The following theorem establishes the equivalency in
observability between the two specific observation matrices. 
\begin{thm}
\label{thm:observ-L-D}The matrix pair ($A$,$\mathcal{D}^{T}$) is
unobservable if and only if the pair ($A$,$\mathcal{L}$) is unobsrevable. 
\end{thm}
\begin{proof}
Suppose that matrix pair $(A,\mathcal{D}^{T})$ is unobservable, then
from Theorem \ref{thm:right_eigen} there is an eigenvector $v$ of
$A$ such that
\[
\mathcal{D}^{T}v=0\Rightarrow\mathcal{D}\mathcal{D}^{T}v=0\Rightarrow\mathcal{L}v=0.
\]
Consequently, we have that $v$ also belongs to the kernel of $\mathcal{L}$,
and it follows that the pair ($A$,$\mathcal{L}$) is unobservable.
For the sufficiency we use the identical argument: suppose that $(A,\mathcal{L})$
is unobservable; there is an eigenvector $v$ of $A$ for which we
can write
\begin{align*}
\mathcal{L}v=0\Rightarrow\mathcal{D}\mathcal{D}^{T}v=0 & \Rightarrow v^{T}\mathcal{D}\mathcal{D}^{T}v=0\\
 & \Rightarrow\left(\mathcal{D}^{T}v\right)^{T}\mathcal{D}^{T}v=0\\
 & \Rightarrow||\mathcal{D}^{T}v||=0.
\end{align*}
This gives $\mathcal{D}^{T}v=0$, and $v$ belongs to the
kernel of $\mathcal{D}^{T}$, which yields the matrix pair ($A$,$\mathcal{D}^{T}$)
unobservable. 
\end{proof}
The last theorem states that the observability of an LTI system monitored
by an RSS is all the same if the observation matrix is the relative
observation matrix or the Laplacian of the relative observation graph. 

To sum up, in this section we examine the observability of dynamical
systems monitored by an RSS. For a connected RSS graph, Theorem \ref{thm:observability_LTI_Dt}
declares that observability is achieved if the eigenvectors of the
state matrix do not belong to $span(\fatone)$, while Lemma \ref{lem:observability_Dt_en_generic}
postulates that the addition of even a single anchor node to the collection
of relative measurements of the RSS will render the LTI system observable.
In the following section, we delve into the impact of a single observation
anchor node to the stability of a special class of LTI systems: the
continuous-time agreement dynamics. 

\section{Observability Analysis of the Single Measurement Agreement Dynamics\label{sec:SingleMeas}}
This section deals with the observability property of coordinated
multi-agent dynamic systems that execute the consensus protocol and
only a single anchor node is observed. In this direction, we consider
the agreement dynamics of a multi-agent system, and, without loss
of generality, we further assume that the anchor measurement is the
state $x_{n}$. Thus, the system under consideration has the following
form:
\begin{align}
\dot{x} & =-\mathcal{L}x\label{eq:concensus-single-anchor}\\
y & =e_{n}x.\nonumber 
\end{align}

This problem can be considered as the observability equivalent of
the controllability of single-leader multi-agent systems. In the latter
instance, the goal is to control the entire multi-agent system by
a single control unit; while in the former case, the objective is
to determine the multi-agent system's state by a single measurement.
The duality of the two properties (observability and controllability)
for the multi-agent system (\ref{eq:concensus-single-anchor}) counters
the rest classes of systems presented in this paper, where the study
of observability cast a unique application domain without an evident
controllability counterpart. 

The controllability property of single-leader multi-agent systems
has been rigorously investigated in \cite{tanner2004controllability,rahmani2009controllability,martini2010controllability}.
An element of particular interest is the effect of the leading node's
placement in the graph to the controllability of the agreement dynamics.
In \cite{rahmani2009controllability} it was shown that leader-symmetry
makes the agreement dynamics uncontrollable. Here, we provide an alternative
to \cite{rahmani2009controllability} proof of this statement, where
we show the relation of the symmetry of the graph and the observability
of the single-measurement multi-agent system (\ref{eq:concensus-single-anchor}).

The following Lemma is the single-measurement observability equivalent
of the proposition presented in \cite{rahmani2009controllability}
for the controllability of single-leader agreement dynamics; the derivation
of its proof is repeated for completeness: 
\begin{lem}[{Observability equivalent of \cite[Proposition 5.4]{rahmani2009controllability}}]
\label{lma:zero_component} The single-measurement agreement dynamics
(\ref{eq:concensus-single-anchor}) are unobservable if and only if
at least one of the eigenvectors of $\mathcal{L}$ has a zero component
on the index that corresponds to the measured (anchor) state. 
\end{lem}
\begin{proof}
Let $\nu_{j}=\left[\nu_{j}^{1};\nu_{j}^{2};\cdots;\nu_{j}^{n}\right]$,
with $j=1,\ldots,n$, be the eigenvector of the Laplacian matrix $\mathcal{L}$.
Theorem \ref{thm:right_eigen} states that the single-measurement
agreement dynamics (\ref{eq:concensus-single-anchor}), with observation
matrix the row vector $e_{n}$, are unobservable if and only if there
exists an eigenvector $\nu_{j}$ of $\mathcal{L}$ such that 
\[
e_{n}\cdot\nu_{j}=0\Rightarrow\nu_{j}^{n}=0.
\]
Thus, if $\mathcal{L}$ has an eigenvector with a zero component on
the index that corresponds to the measured node, the system (\ref{eq:concensus-single-anchor})
is unobservable (sufficiency), and vice versa, i.e., if the system
is unobservable, then $\mathcal{L}$ has at least one eigenvector
with a zero component at its $n^{\text{th}}$ entry (necessity). 
\end{proof}
Lemma \ref{lma:zero_component} is a preparatory step to associate
the observability of the single-measurement agreement dynamics with
the spectral properties of their information-exchange network, encapsulated
by the graph $\mathcal{G}$. The succeeding propositions determine
the observability of the system depending on the location of the anchor
node at the graph. In particular, we intent to establish the relation
between the observability of the system and the symmetry of the graph.
The main result is built gradually, after a sequence of elementary
propositions related to the observability of the system and the spectral
properties of its graph $\mathcal{G}$. The following theorem connects
the observability of system with the eigenvalues of $\mathcal{L}$.
Its original derivation ---in the dual context of controllability
of the single-leader agreement dynamics--- is presented in \cite{rahmani2009controllability}. 
\begin{thm}[{Observability equivalent of \cite[Proposition 5.1]{rahmani2009controllability}}]
\label{thm:eigenvalues} The single-measurement agreement dynamics
(\ref{eq:concensus-single-anchor}) are unobservable if the Laplacian
matrix $\mathcal{L}$ does not have distinct eigenvalues. 
\end{thm}
\begin{proof}
Without loss of generality, assume that $\nu_{1}$ and $\nu_{2}$
are eigenvectors of the eigenvalue $\lambda$ that has multiplicity
two. It is further assumed that the components of $\nu_{1}$ and $\nu_{2}$
that correspond to the index of the anchor node are nonzero; otherwise,
by virtue of Lemma \ref{lma:zero_component}, the system would be
directly unobservable. Because $\mathcal{L}$ is symmetric, it poses
$n$ orthogonal eigenvectors; thus, for the two eigenvectors $\nu_{1}$
and $\nu_{2}$ of $\lambda$, we have $\nu_{1}^{T}\cdot\nu_{2}=0$.
Consider the vector $\nu=\nu_{1}+c\nu_{2}$ with $c\in\real$ that,
by construction, belongs to the eigenspace of $\mathcal{\lambda}$.
It is easy to show that $\nu$ is also an eigenvector of $\mathcal{L}$
because $A\nu=\lambda\nu$. The real constant is selected as $c=-\nu_{1}^{T}e_{n}/\nu_{2}^{T}e_{n}$.
We have already established, by assumption, that $\nu_{1}^{T}e_{n},\nu_{2}^{T}e_{n}\neq0$;
consequently, the constant $c$ is nonzero. For this selection of
$c$, a quick inspection of $\nu$ reveals that its last component
is zero. Therefore, because $\nu$ belongs to the eigenspace of the
eigenvalue $\lambda$ and its last component is zero, in light of
Theorem \ref{thm:right_eigen}, the single-measurement multi-agent
system is unobservable. 
\end{proof}
We will present the concept of graph symmetry by introducing a series
of definitions and statements.
\begin{defn}
A permutation matrix is a square binary matrix with a single nonzero
element in each row and column.
\end{defn}
The $n$-dimensional permutation matrix $P$, when used to pre-multiply
(or post-multiply) a matrix $Q$, of appropriate dimensions, results
in permuting the rows (or columns) of $Q$. For example, if $\left[P\right]_{ij}=1$,
then the $i^{\text{th}}$ row of $Q$ becomes the $j^{\text{th}}$
row of $PQ$. The permutation matrices are orthogonal: for the $n\times n$
permutation matrix $P$ its inverse exists and $PP^{T}=P^{T}P=I_{n}$.
Similar to \cite{Rahmani2008}, we provide the following definition
for the \emph{observation symmetric anchor nodes} of single-measurement
agreement dynamic systems: 
\begin{defn}
The anchor node $x_{i}$, with $i\in\left\{ 1,\ldots,n\right\} $,
is observation symmetric if the Laplacian matrix, $\mathcal{L}$,
of the graph $\mathcal{G}$ of the single-measurement agreement dynamics
system with $y=e_{i}x$ admits a non-identity permutation matrix $\Psi$,
with $\left[\Psi\right]_{ii}=1$, such that
\[
\Psi\mathcal{L}=\mathcal{L}\Psi.
\]
We call the node asymmetric if it does not admit such permutation.
\end{defn}
\begin{figure}
\noindent \begin{centering}
\includegraphics[scale=0.7]{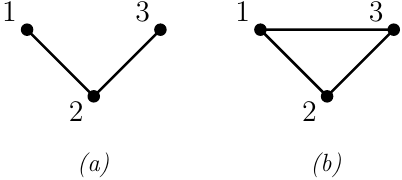}
\par\end{centering}
\caption{Examples of different graph topologies and observation symmetric anchor
nodes: (a) only the anchor node $2$ is observation symmetric; (b)
all the nodes of the graph are observation symmetric. \label{fig:Graph-and-nodes}}
\end{figure}

As an example, node $2$ of the graph illustrated in Figure \ref{fig:Graph-and-nodes}(a)
is observation symmetric and its nodes $1$ and $3$ are asymmetric;
while for the graph of Figure \ref{fig:Graph-and-nodes}(b), all its
nodes are observation symmetric. The following Theorem establishes
the connection between anchor symmetry and the observability property
of single-measurement agreement dynamics systems. 
\begin{thm}
The single-measurement agreement dynamics system (\ref{eq:concensus-single-anchor})
is unobservable if its anchor node is observation symmetric. 
\end{thm}
\begin{proof}
Without loss of generality, we assume that the anchor observation
node is $x_{n}$. It is further assumed that the Laplacian $\mathcal{L}$
has $n$ distinct eigenvalues $\lambda_{i}$, with $i=1,\ldots n$.
If at least one of the eigenvalues has multiplicity greater than one
then, on account of Theorem \ref{thm:eigenvalues}, the multi-agent
system is unobservable. Let $\nu_{i}$ be the eigenvector of $\mathcal{L}$
that corresponds to its eigenvalue $\lambda_{i}$, such that $\mathcal{L}\nu_{i}=\lambda_{i}\nu_{i}$.
Due to the symmetry of matrix $\mathcal{L}$, we can select a set
of orthogonal eigenvectors $\nu_{i}$, with $i=1,\ldots,n$, such
that if $V=\left[\nu_{1},\ldots,\nu_{n}\right]$, we have $VV^{T}=V^{T}V=I$.
Furthermore, because by assumption the anchor node $x_{n}$ is observation
symmetric, there exist a non-identity permutation matrix $\Psi$,
with $\left[\Psi\right]_{nn}=1$, such that $\Psi\mathcal{L}=\mathcal{L}\Psi$.
Using the permutation property of the symmetric graph $\mathcal{G}$,
one has 
\[
\mathcal{L}\left(\Psi\nu_{i}\right)=\left(\mathcal{L}\Psi\right)\nu_{i}=\left(\Psi\mathcal{L}\right)\nu_{i}=\Psi\left(\mathcal{L}\nu_{i}\right)=\lambda_{i}\left(\Psi\nu_{i}\right),
\]
which yields that the vector $\Psi\nu_{i}$ belongs to the eigenspace
of eigenvalue $\lambda_{i}$. Because by assumption ($\mathcal{L}$
has distinct eigenvalues) this subspace is one dimensional; hence,
we infer that $\Psi\nu_{i}$ belongs to the spanning set of the eigenvector
$\nu_{i}$, i.e., $\Psi\nu_{i}\in span(\nu_{i})$. 

An elementary argument can show that the vector $\nu=\nu_{i}-\Psi\nu_{i}$
is also an eigenvector of $\mathcal{L}$ with corresponding eigenvalue
$\lambda_{i}$. More specifically, one has
\begin{align*}
\mathcal{L}\nu & =\mathcal{L}\left(\nu_{i}-\Psi\nu_{i}\right)=\mathcal{L}\nu_{i}-\mathcal{L}\left(\Psi\nu_{i}\right)=\lambda_{i}\nu_{i}-\lambda_{i}\left(\Psi\nu_{i}\right)\\
 & =\lambda_{i}\left(\nu_{i}-\Psi\nu_{i}\right)=\lambda_{i}\nu.
\end{align*}

Recall that for the permutation matrix $\Psi$, we have $\left[\Psi\right]_{nn}=1$;
thus, the $n^{\text{th}}$ component of the eigenvector $\nu$ is
zero, i.e., $\left[\nu\right]_{n}=0$. Consequently, by means of Lemma
\ref{lma:zero_component}, the system is unobservable. 

The same conclusion can be reached with a relatively lengthier succession
of conditionals. We will show this alternate statement for the sake
of completeness. Because the vector $\Psi\nu_{i}$ belongs to the
spanning set of $\nu_{i}$, there exists a real number $l_{i}$ such
that 
\begin{equation}
\Psi\nu_{i}=l_{i}\nu_{i},\label{eq:psi_eigen-1}
\end{equation}
which implies that $\nu_{i}$ is also an eigenvector of $\Psi$ with
corresponding eigenvalue $l_{i}$. Multiplying both sides of (\ref{eq:psi_eigen-1})
with $\nu_{i}^{T}\Psi^{T}=l_{i}\nu_{i}^{T},$ one has 
\[
\left(\nu_{i}^{T}\Psi^{T}\right)\Psi\nu_{i}=l_{i}^{2}\nu_{i}^{T}\nu_{i}\Rightarrow l_{i}^{2}=1.
\]
Therefore, the eigenvalue $l_{i}$ of $\Psi$ can take the values
$\left\{ -1,1\right\} $. If $l_{i}=1$, for all $i\in\{1,\ldots,n\}$,
and taking into consideration that matrix $V$ is orthogonal, we can
write
\[
\Psi\left[\nu_{1},\ldots,\nu_{n}\right]=\left[\nu_{1},\ldots,\nu_{n}\right]\Rightarrow\Psi VV^{T}=VV^{T}\Rightarrow\Psi=I_{n},
\]
which contradicts the initial assumption that $\Psi$ is
a non-identity matrix. Thereby, at least one eigenvalue of $\Psi$
is $-1$, e.g., $l_{i}=-1$, and we have $\Psi\nu_{i}=-\nu_{i}$.
Similarly with before, because $\left[\Psi\right]_{nn}=1$, the $n^{\text{th}}$
component of the eigenvector $\nu_{i}$ is zero; thus, by virtue of
Lemma \ref{lma:zero_component}, the system is unobservable. 
\end{proof}
The concluding remark is that the selection of the anchor node of
a single-measurement agreement dynamics system should be judicious
such that it does not evoke graph symmetry. Furthermore, it is a fact
that a highly connected graph will result to a faster convergence
of the agreement dynamics \cite{mesbahi2010graph}; however, high
connectivity also increases the prospect for observation symmetry
with respect to the measurement (anchor) node.

\section{Design of a Distributed Estimator for the Single-Integrator Multi-Agent
System\label{sec:Distributed-Estimators}}

This section is concerned with the analysis and design of a \emph{distributed
estimator} for the Single-Integrator Multi-Agent (SIMA) system that
is measured by an RSS. A conspicuous application of this case study
is the localization of a group of robots that moves in space. In
this case at hand, each robot (agent) can only measure its relative
location with respect to its adjacent robots, and only one robot can
measure its own global location ---if each state of the system is
perceived as the location of robot in space. 

The observer is tasked to output, in a distributed formalism, an estimate
of the SIMA system's state that converges asymptotically to its actual
value. Each agent's dynamic equation is given by 
\begin{equation}
\dot{x}_{i}=u_{i},\text{ with }i=1,\ldots,n,\label{eq:single-integrator}
\end{equation}
where $x_{i}\in\real$ and $u_{i}\in\real$ is the state
and control input of agent $i$, respectively. For brevity, the single-integrator
dynamics are written as $\dot{x}=u$. In this case study, the agents
can exchange information with other agents. The type of information
that they can exchange will be determined later. The topology of the
information-exchange links is captured by a connected graph $\mathcal{G}$.
We firstly address the observability of the system when the relative
measurements of its observation vector are obtained by an RSS corresponding
to graph $\mathcal{G}$ with Laplacian matrix $\mathcal{L}$. 
\begin{lem}
The SIMA state-space system
\begin{align}
\dot{x} & =u\label{eq:multi-agent-integrator}\\
y & =\mathcal{L}x,\nonumber 
\end{align}
is unobservable.
\end{lem}
\begin{proof}
The SIMA system is a special case of the LTI system $(\ref{eq:LTI-1})$
with $A=O_{n}$, $B=I_{n}$, and $C=\mathcal{L}$. The proof starts
with the observation that $\lambda=0$ is the only eigenvalue of $O_{n}$
with eigenspace the entire $\real^{n}$. Indeed, for every $\nu\in\real^{n}$
it evidently holds $O_{n}\nu=0\nu$. Consequently, the kernel of $\mathcal{L}$,
namely, $span(\fatone_{n})$, is contained in the eigenspace of $O_{n},$
$\real^{n}$. From Theorem \ref{thm:observability_LTI_Dt}, it follows
that the system in (\ref{eq:multi-agent-integrator}) is unobservable.
\end{proof}
According to Lemma \ref{lem:observability_Dt_en_generic}, observability
of the SIMA system monitored by an RSS is warranted with an anchor
node added to its observation vector. With this note, the development
of a distributed observer for the SIMA system follows. Specifically,
we show the procedure to obtain an asymptotic state estimate of the
system when each agent has available relative state measurements and
state estimates of its neighboring nodes, and only a single agent
has available a global (anchor) measurement of its state. 

Without restriction of generality, the state $x_{n}$ is assumed to
be the anchor measurement. The governing dynamics of the \emph{local
observer} for agent $i\in\left\{ 1,\ldots,n-1\right\} $ are
\begin{align}
\dot{\hat{x}}_{i} & =u_{i}-\sum_{j\in\mathcal{N}_{i}}w_{i}^{j}\left[y_{i}^{j}-\left(\hat{x}_{j}-\hat{x}_{i}\right)\right],\label{eq:local-estimator-primary-formulation}
\end{align}
where $\mathcal{N}_{i}$ is the set of adjacent agents to
$i$ determined by the graph $\mathcal{G}$; $\hat{x}_{i}$ is the
local observer's state estimate; $y_{i}^{j}=x_{j}-x_{i}$ is the relative
measurement of agent $j$ with respect to $i$; and $w_{i}^{j}$ is
a positive feedback gain, with $w_{i}^{j}=w_{j}^{i}$ and the pair
$(i,j)$ belonging to the edge set, $\mathcal{E}$, of $\mathcal{G}$.
The edge weights (or gains) are introduced to expedite the convergence
of the estimator dynamics.  Figure \ref{fig:info-sharing} graphically
illustrates the information that an agent exchanges. As it can be
seen from (\ref{eq:local-estimator-primary-formulation}), the global
state, $x_{i}$, does not have a stand-alone presence in the local
observer of agent $i$; rather, feedback terms for the local observer
are the relative measurements, $y_{i}^{j}$, and relative measurements
estimates, $\hat{x}_{j}-\hat{x}_{i}$, with $j\in\mathcal{N}_{i}$.

Recall that that the state $x_{n}$ is an anchor measurement for agent
$n$; hence, the dynamics of the last local observer are given by
\[
\dot{\hat{x}}_{n}=u_{n}-\sum_{j\in\mathcal{N}_{n}}w_{n}^{j}\left[y_{n}^{j}-\left(\hat{x}_{j}-\hat{x}_{n}\right)\right]+K(y_{n}-\hat{x}_{n}),
\]
where $K$ is a positive constant. The last feedback term
is the estimation error of the anchor agent's state. Let $\tilde{x}_{i}=x_{i}-\hat{x}_{i}$
be the error between the actual value of node $i$ and its estimate.
Because $y_{i}^{j}=x_{j}-x_{i}$, the local estimator's dynamics of
(\ref{eq:local-estimator-primary-formulation}) can be written as 
\begin{align}
\dot{\hat{x}}_{i} & =u_{i}-\sum_{j\in\mathcal{N}_{i}}w_{i}^{j}\left(\tilde{x}_{j}-\tilde{x}_{i}\right),\text{ for }i=1,\ldots,n-1;\nonumber \\
\dot{\hat{x}}_{n} & =u_{n}-\sum_{j\in\mathcal{N}_{n}}w_{n}^{j}\left(\tilde{x}_{j}-\tilde{x}_{n}\right)+K\tilde{x}_{n}.\label{eq:estimator-dynamics}
\end{align}

The ``global'' state of the distributed estimator and the estimation
error are denoted by the stack vectors $\hat{x}=\left[\hat{x}_{1};\ldots;\hat{x}_{n}\right]$
and $\tilde{x}=\left[\tilde{x}_{1};\ldots;\tilde{x}_{n}\right]$,
respectively. The dynamics of the distributed estimator are governed
by 
\[
\dot{\hat{x}}=u+\left(\mathcal{L}_{w}+K\Delta_{n}\right)\tilde{x},
\]
and the ``global'' estimator error, $\tilde{x}$, are governed
by
\begin{equation}
\dot{\tilde{x}}=-\left(\mathcal{L}_{w}+K\Delta_{n}\right)\tilde{x},\label{eq:observer-error-dynamics}
\end{equation}
where $\mathcal{L}_{w}=\mathcal{D}W\mathcal{D}^{T}$ is
the weighted Laplace matrix of the RSS graph $\mathcal{G}$, and $\Delta_{n}$
is the $n$-th dimensional diagonal matrix $diag(0,\ldots,0,1)$.
Lastly, $W\in\real^{|\mathcal{E}|\times|\mathcal{E}|}$ is a diagonal
matrix with main diagonal elements the edge weights $w_{i}^{j}$,
and $|\mathcal{E}|$ is the cardinality of the edge set. The orientation
of the observation graph $\mathcal{G}$ does not impact the succeeding
analysis; also, it is not difficult to show that the Laplacian and
the weighted Laplacian of a graph have the same fundamental spaces,
e.g., they have a common kernel. For brevity, we denote the state
matrix of the estimator's error dynamics (\ref{eq:observer-error-dynamics})
by $\Lambda=\mathcal{L}_{w}+K\Delta_{n}$. The following sequence
of statements guarantees the asymptotic convergence of the global
estimator's error.
\begin{lem}
\label{lem:Lambda_eigenvalues_positive}The matrix $\Lambda$ is positive
definite with positive real eigenvalues.
\end{lem}
\begin{proof}
The matrix $\Lambda$ is symmetric as a sum of two symmetric matrices\footnote{It is trivial to show that $\Lambda=\Lambda^{T}$.};
consequently, its eigenvalues are real (see \cite[Property 2; pg. 295]{Strang1988}).
Now, suppose $\lambda$ be an eigenvalue of $\Lambda$ and $\nu=\left[\nu_{1};\ldots;\nu_{n}\right]$
its corresponding eigenvector. To prove that the symmetric matrix
$\Lambda$ is positive definite, it suffices to show that all its
eigenvalues are positive. For any non-zero eigenvector $\nu$ of $\Lambda$,
belonging to its eigenvalue $\lambda$, we have
\[
\nu^{T}\Lambda\nu=\lambda\nu^{T}\nu=\lambda\left\Vert \nu\right\Vert ^{2}\text{; hence, }\lambda=\frac{\nu^{T}\Lambda\nu}{\left\Vert \nu\right\Vert ^{2}}.
\]
Thereby, to prove that $\Lambda$ is a positive definite
matrix, it is sufficient to show that the quadratic form $\nu^{T}\Lambda\nu$
is positive for any eigenvector $\nu$ of $\Lambda$. The quadratic
form $\nu^{T}\Lambda\nu$ can be written as
\begin{align*}
\nu^{T}\Lambda\nu & =\nu^{T}\mathcal{L}_{w}\nu+K\nu^{T}\Delta_{n}\nu\\
 & =\sum_{(i,j)\in\mathcal{E}}w_{i}^{j}\left(\nu_{i}-\nu_{j}\right)^{2}+K\nu_{n}^{2},
\end{align*}
where it is reminded that $\mathcal{E}$ is the edge set
of $\mathcal{G}$. We will show that the two terms in the right-hand
side of the above equality can not be both zero for any eigenvector
$\nu$ of matrix $\Lambda$. Because by assumption the graph $\mathcal{G}$
is connected, the quadratic sum over the edges of $\mathcal{G}$ becomes
zero if and only if the nonzero eigenvector $\nu$ belongs to the
spanning set of $\fatone_{n}$, i.e., $\nu=\in span(\fatone_{n})$.
However, if $\nu\in span\left(\fatone_{n}\right)$, the term $K\nu_{n}^{2}$
is positive; and vice versa, if $\nu_{n}=0$, then $\nu\notin span(\fatone)$
and the quadratic sum over the edges of $\mathcal{G}$ is positive.
Hence, for any eigenvector $\nu$ of $\Lambda$ corresponding to eigenvalue
$\lambda$ the quadratic form $\nu^{T}\Lambda\nu$ is positive, and
$\Lambda$ has strictly positive eigenvalues. 
\end{proof}
The convergence of the global error dynamics of the distributed observer
is stated at the following proposition.
\begin{prop}
\label{prop:Sima-global-error}The global error dynamics of the distributed
observer given in (\ref{eq:observer-error-dynamics}) are asymptotically
stable.
\end{prop}
\begin{proof}
By virtue of Lemma \ref{lem:Lambda_eigenvalues_positive}, the eigenvalues
of the global error dynamics state matrix,$-\Lambda$, of (\ref{eq:observer-error-dynamics})
are real and negative. Therefore, the matrix $-\Lambda$ is Hurwitz
and the dynamics of the global estimation error converge asymptotically
to zero.
\end{proof}
We have detailed the design of a distributed observer for the single-integrator
multi-agent system (\ref{eq:single-integrator}) that renders the
estimation error dynamics asymptotically stable. The available measurements
to each agent are the relative differences between the state variables
of the agent and its neighboring nodes. We have shown, that the stability
of the estimator error dynamics require the availability of at least
one anchor node, namely, a node that can measure the ``global''
value of its state variable.

\begin{figure}
\begin{centering}
\includegraphics[width=0.5\columnwidth]{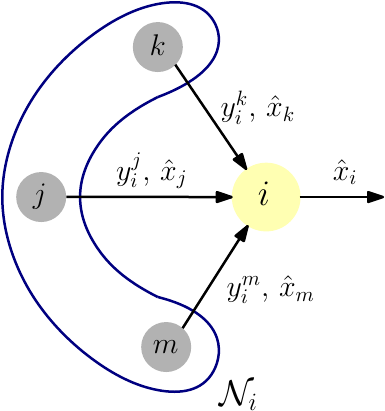}
\par\end{centering}
\caption{Information-sharing links and information-exchanged variables (local
measurements and estimates) between node $i$ and its adjacent nodes
for generating its estimate $\hat{x}_{i}$. \label{fig:info-sharing}}
\end{figure}

\section{Numerical Simulations\label{sec:Simulations}}
In this section, we demonstrate the realization and validity of the
distributed estimator described in Section \ref{sec:Distributed-Estimators}
using numerical simulations. Firstly, we account for a large-scale
system with non-changing dynamics, with state equation $\dot{x}=0$,
and unknown initial conditions. In this facile, non-complex SIMA system,
the input is zero (i.e., $u=0$), and the agent (state) values remain
constant and unknown, equal to their initial condition. 

To showcase our design's performance, regard every agent of the aforementioned
SIMA system as a pixel of a grayscale (black-and-white) image. The
value of each agent ranges from $0$ to $255$, and it represents
the intensity of the pixel: a value of $0$ corresponds to a black
pixel and $255$ to a white. The connectivity of the RSS is determined
by the contiguity of the pixels, that is, every agent ---corresponding
to a pixel--- is adjacent to the agents of its immediate pixels,
in the vertical and horizontal directions, respectively. Ergo, the
agents are connected in such a manner that the relative observation
graph, $\mathcal{G}$, is a grid graph. Recall that the edges of $\mathcal{G}$
represent availability of relative measurements between adjacent agents.
The agent relating to the lower right pixel of the image serves as
the anchor: an anchor node is needed for asymptotic stability of the
monolithic observer error dynamics (see Proposition \ref{prop:Sima-global-error}).
The weights (or gains) of the observer are uniformly selected to be
equal to $1000$, and the weight matrix abbreviates to $W=1000I_{|\mathcal{E}|}$.
The anchor gain, $K$, is also set to $1000$.

For the image that captures the constant agent values of the SIMA
system in point, we use a grayscale, pixelated and downsampled version
of Edward Hopper's Nighthawks painting\footnote{The original colored image, which was used to acquire the simulation
values, in the United States is in the public domain.}, depicted in the rightmost image in the last row of Figure \ref{fig:pixels_state_uno}.
The image resolution is $70\times128$; therefore, the SIMA system
has a total of $8960$ agents. A distributed estimator is designed
according to the guidelines of Section \ref{sec:Distributed-Estimators}.
The initial values of the estimator's state vector are uniformly drawn
from the interval $\left[0\:255\right]$. The double-indexed pixels (indexed
by the row and column number) are simply converted into single-indexed
agents, and the connectivity of the RSS is determined according to
the contiguity of the pixels, as we previously mentioned. The rest
of the images in Figure \ref{fig:pixels_state_uno} show the progressive
convergence of the estimator's state to the actual state of the SIMA
system. The simulation time instance of each frame (viz., state vector)
is depicted on top of every image. We observe that the estimator's
state converges to the actual state in finite time. 

The second simulation case study examines the ability of the observer
to track transient states. Each single-integrator state $x_{i}$ of
(\ref{eq:single-integrator}) tracks the time-varying reference value
$x_{ref}=2\sin\left(5t\right)$, for $t\geq0$, by means of the proportional
control law $u_{i}=-30(x_{i}-x_{ref})$. The distributed estimator's
dynamics are given in (\ref{eq:estimator-dynamics}). The observer's
gain is set to $K=100$. A circle graph was used as the topology of
the sensor network $\mathcal{G}$ that monitors the multi-agent system.
It can be shown that the higher connectivity of $\mathcal{G}$ (larger
$\lambda_{2}$) the faster the convergence of the estimator's dynamics.
A comparison of the actual versus the estimated state in distinct
time instance is given in Figure \ref{fig:pixels_state_uno}. The
actual dynamics exhibit a rapid response, achieving the reference
value $x_{ref}$ substantially faster than the convergence of the
distributed estimator. This observation confirms the inherent time
scale disparity between the system and the observer. The estimated
state and the estimation error with respect to time are shown in Figure
\ref{fig:state_time}.

\begin{figure}[h]
\begin{centering}
\includegraphics[width=0.9\columnwidth]{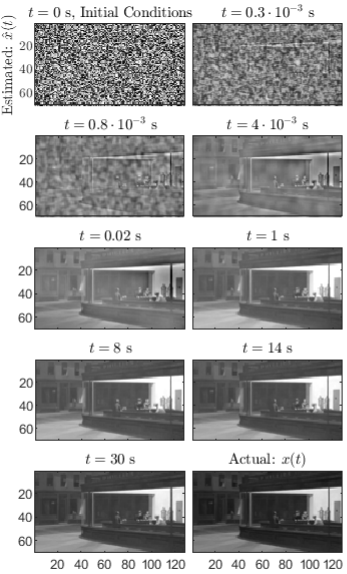}
\par\end{centering}
\caption{Gradual comparison between the actual and estimated states of a SIMA
system at distinct time instances for $n=8960$. Each gray-scale pixel
in the image grid represents the numerical value of a state-space
variable, which maintains a constant value, e.g., $x_{i}(t)=c_{i}$,
where $c_{i}\in\left[0\:255\right]$. The initial state estimates,
$\hat{x}_{i}(0)$, are uniformly drawn from the gray-scale range interval
$\left[0\:255\right]$. \label{fig:pixels_state_uno}}
\end{figure}

\begin{figure}[h]
\begin{centering}
\includegraphics[width=0.9\columnwidth]{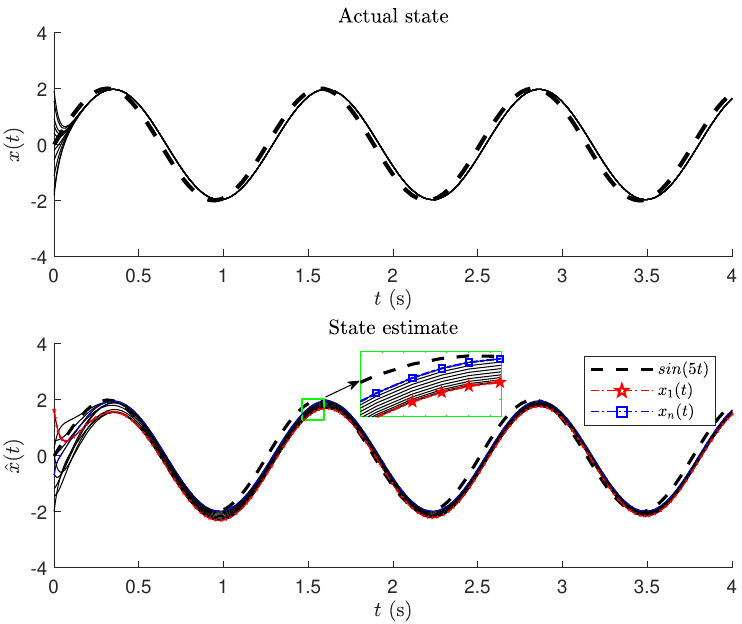}
\par\end{centering}
\caption{Reference time-varying signal (dashed lines), actual states of the
single-integrator (upper plot, solid lines), and estimated states
(lower plot, solid lines), all plotted with respect to time. \label{fig:state_time}}
\end{figure}

\section{Conclusions}

In this paper, we investigate the inferential capacity of RSSs, which
are entities designed to estimate the state of a dynamical system
using relative measurements.

To represent the configuration of relative information availability,
we employ graph theoretical concepts. The matrix representation of
graphs, alongside algebraic graph theory, aligns well with the analysis
of system-theoretic properties in LTI systems. Our approach includes
a vector space, or geometric, perspective to develop a series of verifiable
propositions concerning the observability of LTI systems under RSS
monitoring.

Particular attention is given to a specific category of multi-agent
LTI systems observed by an RSS. We place a special emphasis on consensus
dynamics in two scenarios: one where the structure of the multi-agent
dynamic system and the RSS are congruous, and another where only a
single absolute measurement is available. Ultimately, we design and
examine a distributed observer for a single-integrator system observed
by an RSS. 

\bibliographystyle{IEEEtran}
\bibliography{biblio}

% \begin{IEEEbiographynophoto}{Ioannis Raptis}
% joined the faculty of the Electrical and Computer Engineering Department at North Carolina Agricultural and Technical State University in Fall 2019. He is the director of the Autonomous Robotic Systems Laboratory (ARSL). Dr. Raptis received his Dipl-Ing. in Electrical and Computer Engineering from the Aristotle University of Thessaloniki, Greece, and his Master of Science in Electrical and Computer Engineering from The Ohio State University in 2003 and 2006, respectively. In 2010, he received his Ph.D. degree from the Department of Electrical Engineering at the University of South Florida. 
% \end{IEEEbiographynophoto}

\end{document}